\documentclass[a4paper,UKenglish,cleveref, autoref, thm-restate]{lipics-v2021}

\usepackage{minted}
\usepackage{mathpartir}
\usepackage{tikz}
\usepackage{simplebnf}
\usepackage[ruled]{algorithm2e}
\usetikzlibrary{backgrounds}
\usetikzlibrary{arrows}
\usetikzlibrary{shapes,shapes.geometric,shapes.misc}

\NewDocumentCommand{\TODO}{m m}%
  {{\bfseries\color{#1}[#2]}}%

\RenewDocumentCommand{\TODO}{m m}{}

\newcommand{\Val}{\mathtt{Val}}

\newcommand{\status}[2]{\texttt{status} \; #1 \; #2}
\newcommand{\Idle}{\texttt{Idle}}
\newcommand{\Pending}[2]{\texttt{Pending} \; #1 \; #2}
\newcommand{\Linearized}[1]{\texttt{Linearized} \; #1}

\newcommand{\Prop}{\texttt{Prop}}
\newcommand{\Invoke}[2]{\texttt{Invoke} \; #1 \; #2}
\newcommand{\Response}[1]{\texttt{Response} \; #1}
\newcommand{\Intermediate}{\texttt{Intermediate}}
\newcommand{\Sub}[3]{[#1 / #2]#3}
\newcommand{\EvolveInv}{\texttt{evolve\_inv}}
\newcommand{\EvolveRet}{\texttt{evolve\_ret}}
\newcommand{\LinearizePending}{\mathtt{linearize\_pending}}
\newcommand{\Evolve}{\texttt{evolve}}

\newcommand{\metaconfiguration}[2]{\texttt{meta\_configuration} \; \Pi \; \omega}

\newcommand{\None}[1]{\mathsf{None}}

\newcommand{\Arg}{\mathtt{Arg}}
\newcommand{\TermInvoke}[3]{\mathtt{Invoke}\,#1\,#2\,#3}

\newcommand{\Pair}[2]{\langle #1, #2 \rangle}
\newcommand{\ProjL}[1]{\pi_1(#1)}
\newcommand{\ProjR}[1]{\pi_2(#1)}

\newcommand{\termeval}[7]{\langle #1, #2, #3, #4, #5 \rangle \Downarrow_e \langle #6, #7 \rangle}
\newcommand{\dom}[1]{\text{dom} \; #1}

\newcommand{\Seq}[2]{#1; #2}
\newcommand{\IfStmt}[3]{\mathtt{if} \; #1 \; \mathtt{then} \; #2 \; \mathtt{else} \; #3}
\newcommand{\Assign}[2]{#1 := #2}
\newcommand{\GotoStmt}[1]{\mathtt{goto}\,#1}
\newcommand{\ReturnStmt}[1]{\mathtt{return}\,#1}

\newcommand{\Cont}{\mathsf{Continue}}
\newcommand{\Ret}[1]{\mathsf{Return}\,#1}
\newcommand{\GotoSig}[1]{\mathsf{Goto}\,#1}

\newcommand{\stmteval}[8]{\langle #1,\,#2,\,#3,\,#4,\,#5 \rangle \Downarrow_s \langle #6,\,#7,\,#8 \rangle}
\newcommand{\Initial}[1]{\texttt{Initial} \; #1}
\newcommand{\Step}[4]{\texttt{Step} \; #1 \; #2 \; #3 \; #4}
\newcommand{\final}[1]{\texttt{final} \; #1}

\newcommand{\tracker}[1]{\mathtt{tracker}(#1)}
\newcommand{\RWCell}{\mathsf{RWCELL}}
\newcommand{\OP}[1]{\mathsf{OP}_{#1}}
\newcommand{\Read}{\mathsf{Read}}
\newcommand{\Write}{\mathsf{Write}}
\newcommand{\RCASCell}{\mathsf{RCASCELL}}
\newcommand{\CAS}{\mathsf{CAS}}

\newcommand{\States}[1]{\Sigma_{#1}}
\newcommand{\frameex}[4]{\{ \mathtt{op} = #1, \mathtt{arg} = #2, \mathtt{pc} = #3, \mathtt{registers} = #4\}}
\newcommand{\InitialFrame}[2]{\frameex{#1}{#2}{0}{\emptyset}}
\newcommand{\Next}[1]{\mathsf{Next} \; #1}
\newcommand{\WF}[2]{\mathcal{WF}_{#1}(#2)}

\title{Mechanized Metatheory of Forward Reasoning for End-to-End Linearizability Proofs} 

\author{Zachary Kent}{Carnegie Mellon University}{zkent@cs.cmu.edu}{}{}
\author{Ugur Y. Yavuz}{Boston University}{uyyavuz@bu.edu}{}{}
\author{Siddhartha Jayanti}{Dartmouth University}{jayanti.siddhartha@gmail.com}{}{}
\author{Stephanie Balzer}{Carnegie Mellon University}{balzers@cs.cmu.edu}{}{}
\author{Guy Blelloch}{Carnegie Mellon University}{balzers@cs.cmu.edu}{}{}

\authorrunning{Z. Kent, U.\,Y. Yavuz, S. Jayanti, S. Balzer, and G. Blelloch} 

\Copyright{Zachary Kent, Ugur Y. Yavuz, Siddhartha Jayanti, Stephanie Balzer, and Guy Blelloch} 



\EventEditors{John Q. Open and Joan R. Access}
\EventNoEds{2}
\EventLongTitle{42nd Conference on Very Important Topics (CVIT 2016)}
\EventShortTitle{CVIT 2016}
\EventAcronym{CVIT}
\EventYear{2016}
\EventDate{December 24--27, 2016}
\EventLocation{Little Whinging, United Kingdom}
\EventLogo{}
\SeriesVolume{42}
\ArticleNo{23}

\nolinenumbers

\begin{document}

\maketitle

\begin{abstract}
  In the past decade, many techniques have been developed to prove linearizability, 
  the gold standard of correctness for concurrent data structures. 
  Intuitively, linearizability requires that every operation on a concurrent 
  data structure appears to take place instantaneously, 
  even when interleaved with other operations. 
  Most recently, Jayanti et al. presented the first sound and complete ``forward reasoning'' 
  technique for proving linearizability that relates the behavior of a concurrent data structure 
  to a reference atomic data structure as time moves forward. 
  This technique can be used to produce machine-checked proofs of linearizability in TLA+/TLAPS.
  However, while Jayanti et al.'s approach is shown to be sound and complete, 
  a mechanization of this important metatheoretic result is still outstanding. 
  As a result, it is not possible to produce verified end-to-end proofs of linearizability. 
  To reduce the size of this trusted computing base, we formalize this forward reasoning technique 
  and mechanize proofs of its soundness and completeness in Rocq. 
  As a case study, we use the approach to produce a verified end-to-end proof of linearizability 
  for a simple concurrent register.
\end{abstract}

\section{Introduction}
\label{sec:intro}
Concurrent data structures are invaluable in modern shared-memory parallel computing.
Due to scheduling asynchrony however, even deterministic concurrent algorithms have exponentially many possible executions, thereby making them notoriously difficult to design without subtle errors and race conditions.
To ensure reliability of concurrent data structures, decades of work \cite{DongolDerrick2015} has gone into designing techniques for proving {\em linearizability} \cite{HerlihyWingLinearizability}---the gold standard for correctness, which states that operations, which actually take place over intervals of time, must {\em appear} to take effect atomically.
A significant goal of this line of work has been to design a simple, sound and complete verification technique for producing airtight, machine-verified proofs of linearizability.
In 2024, Jayanti et al. \cite{JJYH2024} took an important stride towards achieving this goal by designing {\em meta-configuration tracking}---the first sound and complete {\em forward reasoning} proof technique for proving linearizability.
This technique is easy to use, since it only requires the prover to reason forwards in time.
Consonantly, it has immediate use in machine-verifying both theoretically difficult and practically relevant concurrent data structures. 
Meta-configuration tracking has been used to machine-verify algorithms that are infamous for being difficult to prove linearizable, such as the Herlihy-Wing queue \cite{HerlihyWingLinearizability}.
Just as importantly, it has been used to machine-verify the correctness of concurrent data structures of considerable practical significance, such as the Jayanti-Tarjan union-find object \cite{JayantiTarjanUF}, which is used in Google's open-source graph mining library to power ``shared memory parallel clustering algorithms which scale to graphs with tens of billions of edges'' \cite{GoogleGraphMining}, and the fast, concurrent Parlay hash table \cite{ParlayHash}.
While important algorithms have been machine-verified via meta-configuration tracking, the technique's {\em metatheory}---i.e., proofs of soundness and completeness---remained un-mechanized.

In this work, we ameliorate this situation by giving the first mechanization of meta-configuration tracking's metatheory.
Furthermore, we show that this enables end-to-end machine-verified proofs of linearizability by proving the correctness of a simple concurrent register algorithm. 
All our proofs are fully mechanized in Rocq.

\subsection{Importance of machine-verifying concurrent data structures}

Due to scheduling asynchrony, even deterministic concurrent algorithms have exponentially many possible executions.
The myriad possibilities make it nearly impossible to ensure race-free concurrent code through conventional unit testing, compelling machine-verification.

The need for machine-verification is also bolstered by several practical examples of subtle races that have had disastrous consequences.
A concurrency race in the power-grid was among the causes of the 2003 Northeast Blackout which effected over 50 million people across the United States and Canada \cite{PoulsenBlackout}.
NASA's Mars Pathfinder repeatedly crashed due to a concurrency bug, nearly jeopardizing the entire mission which required years of preparation \cite{MarsRover}.
Race conditions in the Therac-25 radiation therapy machines caused radiation-poisoning, which eventually resulted in at least three deaths and several more injuries \cite{LevesonTherac}.
In short, fundamental algorithms are deployed in all sorts of systems, including those which are critical to the lives and well-being of people, thus motivating machine-verified rigor in proving subtle concurrent algorithms.

The mechanization of the metatheory that we provide in this work enables such machine-verification guarantees of concurrent data structure correctness.

\subsection{Challenges in proving linearizability and meta-configuration tracking}

Perhaps the most intuitive method for proving linearizability is {\em forward simulation} \cite{JonssonForwardSimulation}.
Linearizability is the condition that operations appear to take effect instantaneously at a point in time in its interval of execution, thus it is often proved by identifying these {\em linearization points}.
If each linearization point can be identified, online, as it occurs in a run, then the proof of linearizability is made simple.
In forward simulation, the prover keeps a copy of an atomic reference object and performs an induction over the steps of an arbitrary run of an algorithm using the implemented object, showing that its behavior is identical to that of the algorithm run with the atomic reference object if the reference object performs operations at the linearization points of the implemented object.

While forward simulation is simple, it can only be used when linearization points can be identified and committed to {\em online}.
Perhaps surprisingly, this turns out to be an important constraint.
Only a small subset of linearizable algorithms, known as {\em strong linearizable} algorithms \cite{GolabStrongLinearizability}, are amenable to this type of proof.
The remaining algorithms exhibit {\em future-dependent} linearizability \cite{DongolDerrick2015}, 
i.e., their linearization points depend on the future (since they cannot be identified online).
Several proofs techniques for such algorithms require more complex reasoning, such as reasoning backwards in time \cite{JonssonBackwardSimulation} or prophesying the future \cite{AbadiLamport1991}.
In contrast, meta-configuration tracking provides a simple solution that essentially boils down to maintaining multiple forward simulations.

The main observation of meta-configuration tracking is that at any point in time $t$, a future-dependently linearizable algorithm has not one, but several different {\em sets} of linearization points.
Each set corresponds to a single forward simulation with an atomic reference object.
As the future unfolds, {\em some} forward simulation will necessarily extend into the future, but determining {\em which one} is not possible online.
In other words, it is not possible to {\em commit} to any single forward simulation in an online manner.
Meta-configuration tracking solves this problem by not committing. 
It just maintains all possible forward simulations over the course of the run and discards a simulation if and when that simulation fails to extend to the future. As long as at least one forward simulation remains, the run is linearizable.
In their paper, Jayanti et al. \cite{JJYH2024} present this method and prove that it is complete---i.e., powerful enough to prove {\em any} linearizable algorithm correct, including the complex future-dependent ones.

Upon mechanizing the metatheory of meta-configuration tracking, we demonstrate the use of the metatheory by giving an end-to-end mechanized proof of a linearizable concurrent register implementation.
The register implementation's code is brief, but it has a future-dependent linearization structure, 
thereby showcasing the full power of the meta-configuration tracking method.

\subsection{Our contributions}

We make the following contributions in this paper:
\begin{enumerate}
	\item
	{\bf Meta-theory mechanization}:
	We mechanize the metatheory of meta-configuration tracking---which is the first complete forward-reasoning proof technique for linearizability.

	\item
	{\bf End-to-end linearizability proof example}:
	We employ our meta-theorem to complete an end-to-end proof of linearizability of a concurrent register algorithm.
	The register implementation is simple, but it is future-dependently linearizable, thereby demonstrating the full power of the meta-configuration tracking approach.
\end{enumerate}
All our proofs are fully mechanized in Rocq.

\subsection{Related work}
\label{sec:related}
In the three decades since linearizability was first introduced by Herlihy and Wing \cite{HerlihyWingLinearizability},
  an extensive body of work has emerged on the topic of techniques and tools for proving linearizability.
For a comprehensive overview of these techniques, 
  we refer the reader to the survey by Dongol and Derrick \cite{DongolDerrick2015}.
Some explored approaches include 
  forward simulation-based techniques \cite{Abdulla2017,Amit2007,Vafeiadis2009,Bouajjani2017}, 
    which rely on forward simulation \cite{JonssonForwardSimulation};
  backward simulation-based techniques \cite{Schellhorn2014,VafeiadisPhD} 
    relying on backward simulation \cite{JonssonBackwardSimulation};
  methods leveraging history and prophecy variables \cite{AbadiLamport1991},
    such as proofs based on \emph{logical atomicity} \cite{Jung2020Prophecy}
    which we will address later in this section;
  aspect-oriented proofs \cite{ChakrabortyHenzinger2015,OhmanNanevski2022,Dodds2015},
    first introduced by Henzinger et al. \cite{HenzingerSezginVafeiadis2013};
  and partial-order maintenance \cite{Khyzha2017}
    which later inspired Oliveira Vale et al.'s recent work 
    on a game semantics-based formulation of linearizability \cite{OliveiraVale2023}.
More recently, the \emph{meta-configuration tracking} of Jayanti et al. 
  has emerged as a new, forward reasoning-based verification technique \cite{JJYH2024}.
These methods vary in their applicability---some are limited to specific 
classes of data structures and hence are not \emph{universal}
(e.g., aspect-oriented proofs).
Others, such as partial-order maintenance \cite{Khyzha2017},
are \emph{incomplete} \cite{OliveiraVale2023,JJYH2024}
meaning they cannot establish linearizability for all correct implementations;
or whether they are \emph{complete} is still an open question \cite{Jung2020Prophecy, OliveiraVale2023}.
Furthermore, they vary in their level of mechanization, with some being fully or partially machine-verified,
while others rely entirely on manual reasoning.

Among these, only two methods are known to be both \emph{universal} and \emph{complete}.
The backward simulation-based approach of Schellhorn et al. \cite{Schellhorn2014} satisfies these properties,
  with both its meta-theoretical soundness and completeness proofs,
  and linearizability proofs for case studies
  fully implemented in the KIV theorem prover \cite{Reif1998}.
However, backward simulation proofs are known to be complex and rather unintuitive for algorithm designers,
  as they require incrementally relating the behavior of an implementation to that of an abstract specification
  in a direction opposite to its execution \cite{VafeiadisPhD, DongolDerrick2015, Khyzha2017}.
In contrast, the meta-configuration tracking technique by Jayanti et al. \cite{JJYH2024} is
  a more recent method that is also universal and complete.
  Crucially, it is the first \emph{forward reasoning} proof technique that avoids 
  backward reasoning and the need for future predictions altogether, making it more intuitive for verification.
  However, while its case studies were fully machine-verified in TLAPS \cite{CousineauTLAPS}, 
  its meta-theoretical results had not been mechanized until our work.

More recently, the concept of \emph{logical atomicity} \cite{DaRochaPinto2014} 
has been introduced as an alternative to linearizability within the framework 
of concurrent separation logic (CSL).
It has gained particular prominence in the higher-order CSL framework Iris \cite{Iris}.
Birkedal et al. \cite{BirkedalLogicalAtomicityLinearizability}
proved that if the CSL specification of a concurrent object operation is logically atomic,
then the operation is necessarily linearizable---a 
meta-theoretical result they formally verified in Rocq using Iris.
Moreover, several concurrent data structures,
including the Herlihy-Wing queue \cite{HerlihyWingLinearizability},
a well-known example of a future-dependent linearizable data structure,
have been shown to be logically atomic in Iris \cite{Jung2020Prophecy}.
However, the completeness of this approach remains an open question,
as it is not known whether every linearizable data structure 
admits a logically atomic CSL specification \cite{OliveiraVale2023}.

\section{Preliminaries}
\label{sec:prelims}


Before discussing our work, we first discuss some basic preliminaries useful for understanding
the intuitive semantics of concurrent execution. A concurrent system is described by a set of processes $\Pi$ that maintain local state and communicate through {\it shared objects} $\Omega$.
When an operation is invoked on a shared object, its internal implementation may in turn perform operations on a set of primitive
\textit{base objects}. Whereas shared objects and their types define an API according to which processes may communicate, base objects concern their low-level implementation details.
For example, consider a concurrent queue backed by an array. An invocation of an enqueue may in turn trigger a write to the array, one of the queue's base objects.
Because these operations are invoked concurrently, the execution of one operation may be simultaneously interleaved with another. This begs the question of how the correctness of a concurrent data structure
ought to be characterized.

The gold standard for correctness of
concurrent data structures is linearizability, which dictates intuitively that every operation on a concurrent
data structure should appear atomic, even when interleaved with other operations. Herlihy and Wing \cite{HerlihyWingLinearizability}
originally defined linearizability with respect to the set of {\it histories} an implementation admits.
A history is a sequence of \textit{invocation} and \textit{response} events corresponding to the call to and return from different operations. A {\it sequential} history is one
in which operations are not interleaved; every invocation is followed immediately by a matching response.
A data structure is specified by a set of valid sequential histories; for example, a sequential history
in which an enqueue of 5 is followed by a dequeue of 6 is invalid. An {\it extension} of a history $A$
is one obtained by appending some number of responses to pending operations in $A$. The {\it completion}
of a history $A$, denoted $\mathsf{complete}(A)$, is the subsequence of $A$ excluding unmatched invocations. 
History $A$ is {\it equivalent} to history $B$ iff $B$ can be obtained by swapping events corresponding to overlapping (concurrent) operations.
Finally, an implementation is {\it linearizable} iff every valid concurrent history $A$ admitted by 
the implementation has an extension whose completion is equivalent to a valid sequential history.

\begin{figure}[!htbp]
  \centering
  \begin{tikzpicture}
    \node[draw, rounded corners=2mm, fill=gray!20, 
      minimum width=0.8cm, minimum height=0.8cm]
      at (0.5, 3.5) {$p_1$};
    
    \node[draw, rounded corners=1mm, fill=green!20, 
      minimum width=0.5cm, minimum height=0.5cm, rotate=45] 
      at (2, 3.5) (1E5) {};
    \node at (2, 4) {{\footnotesize \textsf{Enq}(5)}};
    
    \node[draw, rounded corners=1mm, fill=green!20, 
      minimum width=0.5cm, minimum height=0.5cm, rotate=45] 
      at (6, 3.5) (1E5F) {};
    \draw[thick] ([xshift=0.3cm]1E5.center) -- ([xshift=-0.3cm]1E5F.center);

    \node[draw, circle, fill=violet!50, 
      minimum width=0.5cm, minimum height=0.5cm] 
      at (5.2, 3.5) (1E5L) {};

    \node[draw, rounded corners=2mm, fill=gray!20, 
      minimum width=0.8cm, minimum height=0.8cm]
      at (0.5, 2) {$p_2$};

    \node[draw, rounded corners=1mm, fill=green!20, 
      minimum width=0.5cm, minimum height=0.5cm, rotate=45] 
      at (3, 2) (2E10) {};
    \node at (3, 2.5) {{\footnotesize \textsf{Enq}(10)}};

    \node[draw, rounded corners=1mm, fill=green!20, 
      minimum width=0.5cm, minimum height=0.5cm, rotate=45] 
      at (5, 2) (2E10F) {};
    \draw[thick] ([xshift=0.3cm]2E10.center) -- ([xshift=-0.3cm]2E10F.center);

    \node[draw, circle, fill=violet!50, 
      minimum width=0.5cm, minimum height=0.5cm] 
      at (4, 2) (2E10L) {};

    \node[draw, rounded corners=1mm, fill=green!20, 
      minimum width=0.5cm, minimum height=0.5cm, rotate=45] 
      at (7.5, 2) (2D) {};
    \node at (7.5, 2.5) {{\footnotesize \textsf{Deq}()}};

    \node[draw, rounded corners=1mm, fill=green!20, 
      minimum width=0.5cm, minimum height=0.5cm, rotate=45] 
      at (10.5, 2) (2DF) {};
    \node at (10.5, 2) {{\small $5$}};
    \draw[thick] ([xshift=0.3cm]2D.center) -- ([xshift=-0.3cm]2DF.center);

    \node[draw, circle, fill=violet!50, 
      minimum width=0.5cm, minimum height=0.5cm] 
      at (9.5, 2) (2DL) {};

    \node[draw, rounded corners=2mm, fill=gray!20, 
      minimum width=0.8cm, minimum height=0.8cm]
      at (0.5, 0.5) {$p_3$};
    
    \node[draw, rounded corners=1mm, fill=green!20, 
      minimum width=0.5cm, minimum height=0.5cm, rotate=45] 
      at (5.5, 0.5) (3D) {};
    \node at (5.5, 1) {{\footnotesize \textsf{Deq}()}};

    \node[draw, rounded corners=1mm, fill=green!20, 
      minimum width=0.5cm, minimum height=0.5cm, rotate=45] 
      at (8.5, 0.5) (3DF) {};
    \node at (8.5, 0.5) {{\small $10$}};
    \draw[thick] ([xshift=0.3cm]3D.center) -- ([xshift=-0.3cm]3DF.center);

    \node[draw, circle, fill=violet!50, 
      minimum width=0.5cm, minimum height=0.5cm] 
      at (7, 0.5) (3DL) {};

    \draw[thick, ->] (0, -0.5) -- (12, -0.5) node[below] {\footnotesize time};
  \end{tikzpicture}
  \caption{An execution of a concurrent queue involving three processes. 
  Operations are shown as line segments between two squares,
  marking an invocation and a response. 
  Violet circles denote linearization points. 
  Observe that the return values are consistent with 
  this linearization order.}
  \label{fig:concurrent-queue}
\end{figure}
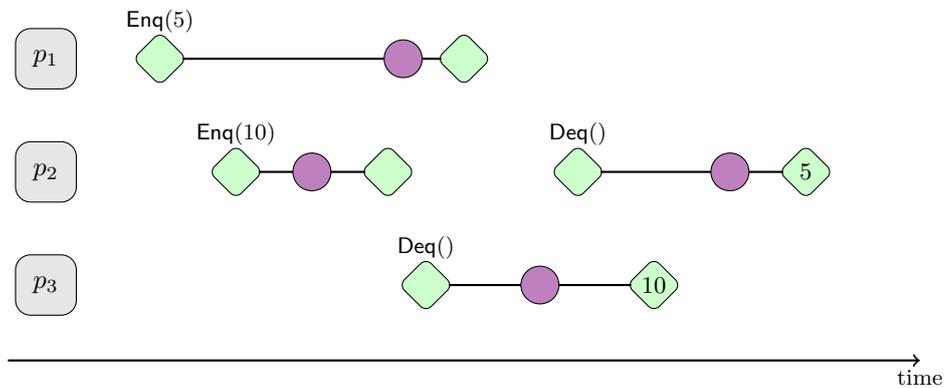

In \Cref{fig:concurrent-queue}, we see a visual depiction of a concurrent execution of a concurrent queue.
Each line segment represents the execution of an operation over time; the left endpoint is its invocation,
while the right endpoint is its response. Some operations are interleaved. If the implementation of the
queue is linearizable, then it must be possible to select a linearization point for every operation
somewhere between its invocation in response. We assume that this is indeed the case, and that
the violet circles on every line are the linearization points.

\section{Model and definitions}
\label{sec:model}
This section lays the foundation upon which the forward reasoning proof technique relies.
It formalizes the components and working of a concurrent system and the notion of linearizability.

\subsection{Components of a concurrent system}

A concurrent system is described by a set of processes $\Pi$ that maintain local state and communicate through {\it shared objects} $\Omega$.
Every shared object is modeled as a state machine, taking transitions from one state to the next when an operation is performed on it.
Each operation on a shared object takes as input (and produces as output) a value drawn from set of values $\Val$. It is given by the following inductive type:

\begin{minted}{coq}
  Inductive Val :=
    | Int (n : Z)
    | Bool (b : bool)
    | Unit
    | Pair (v : Val) (v' : Val).
\end{minted}

The behavior of every object $\omega \in \Omega$ is described by its associated \textit{type}, comprising the following information:

\begin{itemize}
    \item A set $\States{}$ of states
    \item A set $\OP{}$ consisting of the operations that can be performed on it
    \item A transition relation $\delta \subseteq \States{} \times \Pi \times \OP{} \times \Val \times \States{} \times \Val$.
          Specifically, $(\sigma, \pi, \mathit{op}, \mathit{arg}, \sigma', \mathit{ret}) \in \delta$ means that when operation $\mathit{op}$ is performed with argument $\mathit{arg}$ on object $\omega$ in initial state $\sigma$, it transitions to a final state $\sigma'$ and produces a return value $ret$.
          Note that this definition allows transitions to be \textit{nondeterministic}.
\end{itemize}

An object type is represented via the following record, parameterized over the set $\Pi$ of processes:

\begin{minted}[escapeinside=~~, mathescape=true]{coq}
  Record object_type (~$\Pi$~ : Type) := {
    ~$\Sigma$~ : Type;
    OP : Type;
    ~$\delta$~ : ~$\Sigma$~ → ~$\Pi$~ → OP → Val → ~$\Sigma$~ → Val → Prop;
  }.
\end{minted}

To formalize the notion that every object in $\Omega$ ought to be associated with a type, we define a typeclass requiring such:

\begin{minted}[escapeinside=~~]{coq}
  Class Object (~$\Pi$~ : Type) (~$\Omega$~ : Type) `{EqDecision ~$\Omega$~} := type : ~$\Omega$~ → object_type ~$\Pi$~.
\end{minted}

We write $\States{\omega}$ and $\OP{\omega}$ to mean the the set of states and operations associated with the type of $\omega$, respectively.

Now that we have formalized shared objects, we discuss how processes can mutate them (and their own local state).
Every $\pi \in \Pi$ maintains a set of private variables with names drawn from set $\mathit{Var}$ alongside a distinguished program counter. 
An {\it algorithm} specifies an initial state for every object $\omega \in \Omega$ alongside a program that every process $\pi \in \Pi$ executes.
The execution of an algorithm proceeds in atomic steps. During an asynchronous execution, a scheduler chooses some process $\pi$ to take a step at every discrete time step.
When process $\pi$ takes a step, it will execute a \textit{line} of its program. There are three kinds of lines:

\begin{itemize}
    \item \textit{Invocation} lines, corresponding to the invocation of some operation $\mathit{op} \in \OP{\omega}$ on object $\omega$ with some argument $v$.
    \item \textit{Intermediate} lines, corresponding to a line of code in a procedure
    \item \textit{Response} lines, corresponding to the response of an operation, tagged with return value $v$.
\end{itemize}

These lines are represented by the following inductive type

\begin{minted}[escapeinside=~~]{coq}
  Variant line ~$\Pi$~ {~$\Omega$~} `{Object ~$\Pi$~ ~$\Omega$~} (~$\omega$~ : ~$\Omega$~) :=
    | Invoke (op : ~$\OP{\omega}$~) (arg : Val)
    | Intermediate
    | Response (resp : Val).
\end{minted}

We refer to the execution of line $l$ by process $\pi$ as an {\it event}. When the scheduler selects a process to execute one of these lines, the state of the system will transition from 
one configuration to the next. Each configuration describes a snapshot of the system's state
at a given point in time. A \textit{run} denotes one possible sequence of alternating configurations and events, beginning in an initial configuration.
We encode it as follows:

\begin{minted}[escapeinside=~~]{coq}
  Inductive run (C ~$\Pi$~ : Type) `{Object ~$\Pi$~ ~$\Omega$~} (~$\omega$~ : ~$\Omega$~) :=
    | Initial (c : C)
    | Step (r : run C ~$\Pi$~ ~$\Omega$~) (~$\Pi$~ : ~$\Pi$~) (l : line ~$\Pi$~ ~$\Omega$~) (c : C).
\end{minted}

This type is parameterized over the type of configuration $\mathtt{C}$, which is left abstract for now. In the simplest case,
a configuration maintains a global state assignment for every object and local state assignment for every process.
However, as we will see, some configurations also maintain additional auxiliary proof state.



\subsection{Implementations using base objects}

\begin{figure}[htbp]
\begin{bnf}
  $e$ : \textsf{Term} ::= $x$ // $n \in \mathbb{N}$ // $b \in \mathbb{B}$ // $\top$ // $\Arg$ 
  | $\ProjL{e}$ // $\ProjR{e}$ // \texttt{$e_1 \oplus e_2$} // $\neg e$ // $\Pair{e_1}{e_2}$ // \texttt{$\TermInvoke{\omega}{\mathit{op}}{\mathit{arg}}$}
  ;;
  $s$ : \textsf{Statement} ::= $\Seq{s}{s}$ // $\Assign{x}{e}$ // $\IfStmt{e}{s_1}{s_2}$ // $\ReturnStmt{e}$ // $\TermInvoke{\omega}{\mathit{op}}{\mathit{arg}}$
\end{bnf}

\caption{Syntax of Terms and Statements}
\label{fig:syntax}

\end{figure}

We now work towards making this semantics concrete; until this point, the syntax of the programs each process
executes, alongside their behavior, has been left abstract. Note that runs describe interleavings of
operations on high-level shared data structures, such as queues. However, abstract data structures are often
implemented using lower-level primitives, such as read-write registers. That is, the intermediate lines
of high-level operations may invoke other operations on low level primitives, a set of \textit{base objects}.
Intermediate lines execute simple imperative \textit{statements} that mutate base objects. Statements may also contain \textit{terms} that evaluate down to values. Their syntax, alongside those of terms, is given by 
the grammar in \Cref{fig:syntax}. Statements include assignments of terms to local variables, invocations of operations on base objects,
\textsf{goto}'s that transfer control to another line, sequential composition, and procedure returns.
Formally, an {\it implementation} of $\omega \in \Omega$ using a set of base objects $\overline{\Omega}$ provides the
following information:

\begin{itemize}
  \item The initial state $\sigma_0 \in \States{\omega}$ of the implemented object
  \item A state assignment for every base object
  \item A {\it procedure} describing how to implement every operation on $\omega$. Procedures are simply lists of statements, each of which may manipulate some base object of $\overline{\Omega}$.
\end{itemize}


Having defined the syntax of implementations, we now describe their dynamics. We give semantics to both
terms and statements via a big-step evaluation dynamics. Evaluating statements can alter both local and shared state. Additionally, evaluating a statement may result
executing a return or goto, and thus this control ought to be reflected in their dynamics. We represent control
transfer explicitly as the following data type:

\begin{minted}[escapeinside=~~]{coq}
Variant signal :=
  | Continue
  | Goto (l : nat)
  | Return (v : Value.t).
\end{minted}

In addition to mutating shared and local state, executing a statement produces a signal that determines
what line will be executed when $\pi$ next takes another step. The dynamics takes the form
$\langle \pi , \mathit{arg} , \psi , \epsilon , s \rangle \Downarrow_s \langle \psi' , \epsilon' , \mathit{sig} \rangle$, where:
\begin{itemize}
  \item $\pi$ is the process evaluating the term
  \item $\mathit{arg}$ is the argument provided to the high-level operation
  \item $\psi : \mathit{Var} \rightharpoonup \Val$ is valuation of the local store of $\pi$
  \item $\epsilon : \forall \omega : \overline{\Omega}, \States{\omega}$ is the initial base object state assignment before evaluating $s$
  \item $s$ is the statement being evaluated 
  \item $\epsilon' : \forall \omega : \overline{\Omega}, \States{\omega}$ is the updated base object state assignment resulting from evaluating $s$
  \item $\mathit{sig}$ is the resulting signal
\end{itemize}

Rules defining the dynamics appear in \Cref{fig:sem:stmt}. We omit a description of the term dynamics here, as they are typical.

\begin{figure}[htbp]
\centering
\def \MathparLineskip {\lineskip=0.75em}
\begin{mathpar}




  \inferrule[Seq-Cont]
    { \stmteval{\pi}{\mathit{arg}}{\psi}{\epsilon}{s_1}{\psi_1}{\epsilon_1}{\Cont} \\
      \stmteval{\pi}{\mathit{arg}}{\psi}{\epsilon_1}{s_2}{\psi_2}{\epsilon_2}{sig} }
    { \stmteval{\pi}{\mathit{arg}}{\psi}{\epsilon}{\Seq{s_1}{s_2}}{\psi_2}{\epsilon_2}{sig} }
\\\\

  \inferrule[If-True]
    { \termeval{\pi}{\mathit{arg}}{\psi}{\epsilon}{e}{\epsilon'}{\mathsf{true}} \\
      \stmteval{\pi}{\mathit{arg}}{\psi}{\epsilon'}{s_1}{\psi_1}{\epsilon_1}{sig} }
    { \stmteval{\pi}{\mathit{arg}}{\psi}{\epsilon}{\IfStmt{e}{s_1}{s_2}}{\psi_1}{\epsilon_1}{sig} }
\\\\


  \inferrule[Assign]
    { \termeval{\pi}{\mathit{arg}}{\psi}{\epsilon}{e}{\epsilon'}{v} }
    { \stmteval{\pi}{\mathit{arg}}{\psi}{\epsilon}{\Assign{x}{e}}{\Sub{v}{x}{\psi}}{\epsilon'}{\Cont} }

  \inferrule[Goto]
    { }
    { \stmteval{\pi}{\mathit{arg}}{\psi}{\epsilon}{\GotoStmt{n}}{\psi}{\epsilon}{\GotoSig{n}} }
\\\\

  \inferrule[Seq-Ret]
    { \stmteval{\pi}{\mathit{arg}}{\psi}{\epsilon}{s_1}{\psi_1}{\epsilon_1}{\Ret{v}} }
    { \stmteval{\pi}{\mathit{arg}}{\psi}{\epsilon}{\Seq{s_1}{s_2}}{\psi_1}{\epsilon_1}{\Ret{v}} }

  \inferrule[Invoke]
    { \termeval{\pi}{\mathit{arg}}{\psi}{\epsilon}{\TermInvoke{\omega}{op}{e}}{\epsilon'}{v} }
    { \stmteval{\pi}{\mathit{arg}}{\psi}{\epsilon}{\TermInvoke{\omega}{op}{e}}{\psi}{\epsilon'}{\Cont} }

\end{mathpar}
\caption{Selected evaluation dynamics for statements}
\label{fig:sem:stmt}

\end{figure}

When executing a procedure implementing an operation, a process maintains its local state within a {\it frame}.
A frame comprises a valuation for the local store, the value of the program counter, the pending operation,
and the argument with which the pending operation was invoked. It is represented by the following type:

\begin{minted}[escapeinside=~~]{coq}
  Record frame ~$\Pi$~ {~$\Omega$~} `{Object ~$\Pi$~ ~$\Omega$~} (~$\omega$~ : ~$\Omega$~) : Type := {
    op : (type ~$\omega$~).(OP);
    pc : ~$\mathbb{N}$~;
    arg : Val;
    registers : Var ~$\rightharpoonup$~ Val;
  }.
\end{minted}

Thus, when process $\pi$ executes the statement pointed to by its program counter, it steps from one frame to another 
(or none, if the statement is a return). We represent this outcome--either a next frame, or a return with value $v$--as the following sum type.

\begin{minted}[escapeinside=~~]{coq}
Variant procedure_signal ~$\Pi$~ {~$\Omega$~} `{Object ~$\Pi$~ ~$\Omega$~} ~$\Omega$~ :=
  | Next (f : frame ~$\Pi$~ ~$\Omega$~) (* On next step, go to line ~$l$~ *)
  | Return (v : Val).
\end{minted}

We now give a big-step dynamics for frames. Given a shared state assignment $\epsilon$, frame $f$
will evaluate down to a new shared state assignment $\epsilon'$ and procedure signal $\mathit{sig}$. The dynamics then
takes the form $\langle \pi, \epsilon, f \rangle \Downarrow_f \langle \epsilon, \mathit{sig} \rangle$.
It is defined by three rules: one for each signal the execution of the line pointed to by the current program counter of $\pi$ may produce. If executing this statement produces
results in $\Cont$, then the program counter is incremented in the resulting frame. If the signal
is $\GotoSig{\mathit{pc'}}$, the program counter is set to $\mathit{pc'}$. And if the signal is $\Ret{v}$, then that
signal is propagated.

\subsection{Global dynamics}

So far, we have only discussed the dynamics of atomically executed statements. We now describe the behavior of
concurrent execution. When the scheduler chooses a process to take a step, it transforms one configuration to another. Having previously only discussed abstract configurations, we now
make them concrete. A configuration comprises both shared global state (a state assignment for every base object)
and the local state valuation for each process. We represent local state as a partial mapping between processes and frames. If a process is associated with
no frame, it has no outstanding operation. It is represented by the following record type:

\begin{minted}[escapeinside=~~]{coq}
  Record configuration ~$\Pi$~ ~$\overline{\Omega}$~ {~$\Omega$~} `{Object ~$\Pi$~ ~$\overline{\Omega}$~, Object ~$\Pi$~ ~$\Omega$~} (~$\omega$~ : ~$\Omega$~) := {
    outstanding : ~$\Pi$~ ~$\rightharpoonup$~ frame ~$\Pi$~ ~$\omega$~;
    ~$\epsilon$~ : ~$\forall \overline{\omega} : \overline{\Omega}$~, ~$\Sigma_{\overline{\omega}}$~;
  }.
\end{minted}

As noted, every process with an outstanding operation is associated with a frame defining the corresponding local state associated with that
operation. Thus, a \textit{configuration}---a snapshot of system-wide state at a discrete time step---comprises both a
partial mapping between processes and frames and a shared-state valuation. 

We now define a small-step dynamics of the form $\langle \chi, \epsilon, \pi, l \rangle \mapsto \langle \chi', \epsilon' \rangle$
that describes how the global shared state $\epsilon$ and local state $\chi$ evolves to $\epsilon'$ and $\chi'$, respectively,
when process $\pi$ executes line $l$. It is defined by the rules in \Cref{fig:sem:conf}. If $\pi$ has no outstanding
operation, it can invoke a new operation. If $\pi$ does have an outstanding operation, it can execute the line of code pointed
to by its program counter. If executing that line returns a value $v$, this yields a response of $v$ from the outstanding operation.
Otherwise, the operation remains outstanding. 

Recall that a run is an alternating sequence of configurations and atomic steps. However, this only characterizes the syntactic nature of a run.
Some runs that are syntactically valid obviously have no meaning; it ought to be the case that each step in the run respects the dynamics we defined.
We consider runs that do respect these dynamics to be \textit{well-formed}. Formally, fixing an implementation \texttt{impl}, we define a property $\WF{\texttt{impl}}{-}$
to be the strongest property $\mathcal{P}$ of runs such both of the following hold:

\begin{enumerate}
  \item $\mathcal{P}(\Initial{c})$ holds iff the shared state valuation of $c$ is equal to the initialization specified by the implementation and the mapping between processes and frames specified by $c$ is empty.
  \item Fix any run $r$ and assume the final configuration of $r$ has a frame mapping $\chi$ and shared state valuation $\epsilon$. Let $c'$ be another configuration with
        frame mapping $\chi'$ and shared state valuation $\epsilon'$. Then, if $\mathcal{P}(r)$ and $\langle \chi, \epsilon, \pi, l \rangle \mapsto \langle \chi', \epsilon' \rangle$, then
        $\mathcal{P}(\Step{r}{\pi}{l}{c'})$ holds as well.
\end{enumerate}

We encode this inductively defined property as an inductively defined propositon in our mechanization.

\begin{figure}[htbp]
  \def \MathparLineskip {\lineskip=0.75em}
  \begin{mathpar}
  
  \inferrule[Step-Invoke]
      { \pi \notin \dom{\chi} }
      { \langle \chi, \epsilon, \pi, \Invoke{\mathit{op}}{\mathit{arg}} \rangle \mapsto \langle \Sub{\InitialFrame{\mathit{op}}{\mathit{arg}}}{\pi}{\chi}, \epsilon \rangle }
  
  \inferrule[Step-Intermediate]
      { x \in \dom{\chi} \\ \chi(x) = f \\ \langle \pi, \epsilon, f \rangle \Downarrow_f \langle \epsilon', \Next{f'} \rangle }
      { \langle \chi, \epsilon, \pi, \Intermediate \rangle \mapsto \langle \Sub{f'}{\pi}{\chi}, \epsilon' \rangle }

  \inferrule[Step-Response]
      { x \in \dom{\chi} \\ \chi(x) = f \\ \langle \pi, \epsilon, f \rangle \Downarrow_f \langle \epsilon', \Ret{v} \rangle }
      { \langle \chi, \epsilon, \pi, \Response{v} \rangle \mapsto \langle \Sub{\bot}{\pi}{\chi}, \epsilon' \rangle }
  \end{mathpar}
  \caption{Global Dynamics}
  \label{fig:sem:conf}
  
  \end{figure}

We now move onto consider the correctness of an implementation. As discussed in \Cref{sec:prelims},
\textit{linearizability} is considered the gold standard for the correctness of concurrent data structures. However, Herlihy and Wing's original definition is not particularly amenable to verification. As a result,
other (provably equivalent) definitions have been proposed. We use Lynch's definition \cite{LynchBOOK1996}, which relates the
behavior of a concurrent object to a reference \textit{atomic}  object as time moves forward.

Formally, the \textit{atomic implementation} of an object $\omega \in \Omega$ takes $\Omega$ as the set of base objects.
The implementation itself is trivial. The procedure assigned to every $\mathit{op} \in \OP{\omega}$ consists of only two lines:

\begin{enumerate}
  \item An invocation of $\mathit{op}$, and an assignment of the result to a return variable $r$
  \item A return of $r$
\end{enumerate}

We call any run of this implementation an \textit{atomic run}. Consider the structure of any such atomic run.
Every outstanding operation will execute exactly one intermediate line between its invocation
and response. This is the \textit{linearization point} of the operation: the point at which it takes effect.
Thus, at any point in this run, every process $\pi \in \Pi$ is in exactly one of three states:

\begin{itemize}
  \item \texttt{Idle}, meaning that that process has no outstanding operation
  \item \texttt{Pending} with operation \texttt{op} and argument \texttt{arg}, meaning that $\pi$ has an outstanding operation that has not yet linearized (i.e. taken effect)
  \item \texttt{Linearized} with value $v$, meaning that $\pi$ has outstanding operation that has linearized with value $v$.
\end{itemize}

We encode this in the following sum type:

\begin{minted}[escapeinside=~~]{coq}
Variant status ~$\Pi$~ {~$\Omega$~} `{Object ~$\Pi$~ ~$\Omega$~} (~$\omega$~ : ~$\Omega$~) :=
  | Idle
  | Pending (op : ~$\OP{\omega}$~) (arg : Val)
  | Linearized (res : Val).
\end{minted}

Thus, at any point in an atomic run, the global state can be described by two components:

\begin{enumerate}
  \item The state $\sigma$ of the implemented object $\omega$
  \item A mapping $f : \Pi \to \status{\Pi}{\omega}$ between processes and their statuses
\end{enumerate}

An \textit{atomic configuration} is such a pair $(\sigma, f)$. A step in an atomic run will then
transform one atomic configuration into another. We now define a dynamics $\leadsto$ formalizing
valid steps in an atomic run such that $\langle \sigma, f, \pi,  l \rangle \leadsto \langle \sigma', f' \rangle$
iff process $\pi$ executing $l$ in atomic configuration $(\sigma, f)$ results in atomic configuration $(\sigma', f')$.
It is defined by the rules in \Cref{fig:sem:atomic}. As expected, there are three rules corresponding to the invocation, linearization, and response of an operation.
If $\pi$ is idle, it can invoke an operation. If $\pi$ is associated with pending operation, it can linearize that operation.
And finally, if $\pi$ has a pending operation that has linearized with value $v$, the operation can respond with value $v$.
As before, we can also consider the well-formedness of atomic runs with respect to these dynamics; well-formed runs are
those whose initial configuration is indeed initial, and whose every step respects these dynamics. Formally, we define a predicate $\WF{\mathtt{atomic}}{-}$ on atomic runs
to be the strongest property $\mathcal{P}$ such that both of the following are true:

\begin{enumerate}
  \item $\mathcal{P}(\Initial{(\sigma, f)})$ is true iff $\sigma = \sigma_0$, the initialization of $\omega$, and $f(\pi) = \Idle$ for all $\pi \in \Pi$.
  \item Consider any run $r$ and let the final configuration of $r$ be $(\sigma, f)$. Then, if $\mathcal{P}(r)$ holds and $\langle \sigma, f, \pi,  l \rangle \leadsto \langle \sigma', f' \rangle$, then $\mathcal{P}(\Step{r}{\pi}{l}{(\sigma', f')})$ holds as well.
\end{enumerate}

\begin{figure}[htbp]
  \def \MathparLineskip {\lineskip=0.75em}
  \begin{mathpar}
  
  \inferrule[Step-Invoke]
      { f(\pi) = \Idle }
      { \langle \sigma, f, \pi, \Invoke{\mathit{op}}{\mathit{arg}} \rangle \leadsto \langle \sigma, \Sub{\Pending{\mathit{op}}{\mathit{arg}}}{\pi}{f}, \rangle }
  
  \inferrule[Step-Intermediate]
      { f(\pi) = \Pending{\mathit{op}}{\mathit{arg}} \\ (\sigma, \pi, \mathit{op}, \mathit{arg}, \sigma', \mathit{v}) \in \delta_\omega  }
      { \langle \sigma, f, \pi, \Intermediate \rangle \leadsto \langle \sigma', \Sub{\Linearized{v}}{\pi}{f} \rangle  }

  \inferrule[Step-Response]
      { f(\pi) = \Linearized{v} }
      { \langle \sigma, f, \pi, \Response{v} \rangle \leadsto \langle \sigma, \Sub{\bot}{\pi}{f}, \rangle }
  \end{mathpar}
  \caption{Semantics for steps taken in atomic runs}
  \label{fig:sem:atomic}
  
  \end{figure}

\subsection{Linearizability}

With these definitions, we can now formalize when an implementation is \textit{linearizable}.
The \textit{behavior} of a run is the subsequence of that run consisting only of its invocation and response events;
it is analogous to a history in Herlihy-Wing linearizability. An implementation is \textit{linearizable} iff
the behavior of every valid run it admits is equal (syntactically) to the behavior of some atomic run.
We refer to this witness as a \textit{linearization}. Intuitively, each intermediate line executed by an operation in the atomic run is the linearization
point of the same operation of the concurrent run. Lynch proves that this definition of linearizability
is equivalent to Herlihy and Wing's original definition \cite{LynchBOOK1996} (Theorems 13.4 and 13.5). To mechanize these ideas, we define
a function \texttt{behavior} that computes the behavior of a run, represented by a list of $(\pi, l)$
pairs that are the events of that run. It simply filters out all intermediate lines. Now, we can mechanize \textit{linearizability} as described informally above:

\begin{minted}[escapeinside=~~]{coq}
  Definition linearizable impl :=
    ~$\forall$~ r, ~$\WF{\mathtt{impl}}{\mathtt{r}}$~ ~$\to$~
      ~$\exists$~ r~$_\texttt{atomic}$~,
        ~$\WF{\mathtt{atomic}}{\mathtt{r}_\mathtt{atomic}}$~ ~$\wedge$~ behavior r = behavior r~$_\texttt{atomic}$~
\end{minted}

\section{Proof technique and metatheory}
\label{sec:tracker}
Now that we have mechanized linearizability itself, we do the same with Jayanti et al.'s technique.
Informally, we will augment every run with "ghost" state that tracks the possible linearizations of that run
but has no effect on the execution of the program. More specifically, a \textit{metaconfiguration} is a
set of atomic configurations, each of which is the final configuration of a linearization of the current run.
If this set is non-empty at every step in every run admitted by an implementation, then there exists
at least one linearization of every run, and by definition the implementation is thus linearizable.

We now work towards mechanizing this proof technique. For the remainder of this section, fix a
set of processes $\Pi$ and an implementation \texttt{impl} for object 
implemented object $\omega \in \Omega$ with a set of base objects $\overline{\Omega}$.
Recall that an atomic configuration is simply a pair $(\sigma, f) \in \States{\omega} \times (\Pi \to \status{\Pi}{\omega})$
consisting of a state of $\omega$ and a mapping between processes and their statuses.
We represent a set of atomic configurations as a relation $C : \States{\omega} \to (\Pi \to \status{\Pi}{\omega}) \to \Prop$
such that $(\sigma, f) \in C$ iff $C \; \sigma \; f$ is derivable.

\begin{figure}[htbp]
  \def \MathparLineskip {\lineskip=0.75em}
  \begin{mathpar}
  
  \inferrule[Refl]
      { }
      { (\sigma, f, \cdot, \sigma, f) \in \delta^\ast_\omega }
  
  \and
  \inferrule[Step]
      { (\sigma, f, \vec{\pi}, \sigma', f') \in \delta^\ast_\omega \\ f'(\pi) = \Pending{\mathit{op}}{\mathit{arg}} \\ (\sigma', \pi, \mathit{op}, \mathit{arg}, \sigma'', \mathit{v}) \in \delta_\omega}
      { (\sigma, f, \vec{\pi} \circ \pi, \sigma'', \Sub{\Linearized{v}}{\pi}{f'}) \in \delta^\ast_\omega }
  \end{mathpar}
  \caption{Multistep Transition Relation}
  \label{fig:multistep}
  
  \end{figure}

We now consider how this ghost state should change upon execution of each line in the program.
The metaconfiguration, at any point in a run $r$, should contain exactly the final configurations
of possible linearizations of $r$. So, we must consider how the set of possible linearizations changes
upon execution of a given line. When process $\pi$ executes line $l$, an arbitrary sequence of processes
with pending operations may linearize directly afterwards. To express this notion, we define a
\textit{multistep} transition relation with signature
$\delta^* \subseteq \States{\omega} \times (\Pi \to \status{\Pi}{\omega}) \times \Pi^\ast \times \States{\omega} \times (\Pi \to \status{\Pi}{\omega})$
such that $(\sigma, f, \vec{\pi}, \sigma', f') \in \delta^\ast$ iff atomic configuration $(\sigma', f')$ results from
linearizing sequence of processes $\vec{\pi}$, starting from atomic configuration $(\sigma, f)$. It is defined
inductively by the rules in \Cref{fig:multistep}. There are two cases:

\begin{itemize}
  \item For the base case, linearizing zero processes does not
  change the atomic configuration; that is, for arbitrary $(\sigma, f)$, we have $(\sigma, f, \cdot, \sigma, f) \in \delta^\ast$.
  \item Inductively, assume that linearizing $\vec{\pi}$ results in atomic configuration
  $(\sigma', f')$, and $\pi$ is outstanding with pending operation $\mathit{op}(\mathit{arg})$. Further, suppose that
  $(\sigma', \pi, op, arg, \sigma'', res) \in \delta_\omega$; that starting in state $\sigma'$, $\pi$ linearizing operation $\mathit{op}$
  with argument $\mathit{arg}$ transitions $\omega$ to state $\sigma'$ and produces result $\mathit{res}$. Then,
  $\pi$ can be linearized directly after $\vec{\pi}$, and we have that $(\sigma, f, \vec{\pi} \circ \pi, \sigma'', \Sub{\Linearized{res}}{\pi}{f'})$.
\end{itemize}

Given this multistep transition relation, we define a function $\LinearizePending$ that maps
one metaconfiguration $C$ to another containing exactly the atomic configurations that can be obtained
by linearizing some number of pending processes, starting from some atomic configuration $(\sigma, f) \in C$.
Formally, we have that $(\sigma', f') \in \LinearizePending \ C$ iff there exists some $(\sigma, f) \in C$
and $\vec{\pi} \in \Pi^\ast$ where $(\sigma, f, \vec{\pi}, \sigma', f') \in \delta^\ast$.

Before an arbitrary sequence of processes are allowed to linearize, however, the execution of line $l$ itself
may refine the set of possible linearizations. For example, if an operation responds with value $v$, then
all atomic configurations reflecting the linearization of that operation with a different value must be discarded.
We consider the effect of each different kind of line: invocations, intermediate lines, and responses.

\begin{itemize}
  \item Case $\Invoke{\textit{op}}{\textit{arg}}$: Every atomic configuration $(\sigma, f)$ where $f(\pi) = \Idle$ should 
        be mapped to $(\sigma, \Sub{\Pending{\textit{op}}{\textit{arg}}}{\pi}{f})$. The state $\sigma$ remains unchanged in the resulting atomic configuration,
        as invocations do not modify the state; the state only changes when operations linearize. We define a function $\EvolveInv$ from metaconfigurations to metaconfigurations respecting this behavior:
        \[
        \EvolveInv(C, \mathit{op}, \mathit{arg}) \triangleq \{ (\sigma, \Sub{\Pending{\mathit{op}}{\mathit{arg}}}{\pi}{f}) \mid (\sigma, f) \in C \wedge f(\pi) = \Idle\}
        \]
  \item Case $\Intermediate$: The execution of intermediate lines themselves does not change the set of possible linearizations.
  \item Case $\Response{v}$: As discussed above, all atomic configurations that reflect the operation linearizing with some $w \neq v$ should be discarded. Those remaining---configurations $(\sigma, f)$ where $f(\pi) = \Linearized{v}$---should be mapped to $(\sigma, \Sub{\pi}{\Idle}{f})$, reflecting that $\pi$ no longer is outstanding. As for the invocation case, we define a function $\EvolveRet$ respecting this behavior:
        \[
        \EvolveRet(C, v) \triangleq \{ (\sigma, \Sub{\Idle}{\pi}{f}) \mid (\sigma, f) \in C \wedge f(\pi) = \Linearized{v}\}
        \]
\end{itemize}

With this specification, we can now define a function $\Evolve$
such that $\Evolve \; \pi \; l \; C$ is the metaconfiguration resulting from process $\pi$ executing line $l$ given current metaconfiguration $C$.
It is defined by cases on line $l$:

\begin{align*}
  \Evolve \; \pi \; (\Invoke{\mathit{op}}{\mathit{arg}}) \; C &\triangleq \LinearizePending(\EvolveInv(C, \mathit{op}, \mathit{arg}))\\
  \Evolve \; \pi \; \Intermediate \; C &\triangleq \LinearizePending(C)\\
  \Evolve \; \pi \; (\Response{v}) \; C &\triangleq \LinearizePending(\EvolveRet(C, v))
\end{align*}

\subsection{Tracking Technique}

Previously, in \Cref{sec:model}, configurations only described the local state for every process, alongside a state assignment for the shared global state.
Now, however, we must also maintain the state of the tracker in the configuration. We encode this in the following type, comprising a
\textit{base configuration} encapsulating the concrete shared and local state, and a metaconfiguration representing the tracker.

\begin{minted}[escapeinside=~~]{coq}
  Record configuration ~$\Pi$~ ~$\overline{\Omega}$~ {~$\Omega$~} `{Object ~$\Pi$~ ~$\overline{\Omega}$~, Object ~$\Pi$~ ~$\Omega$~} (~$\omega$~ : ~$\Omega$~) := {
    base_configuration : Implementation.configuration ~$\Pi$~ ~$\Omega$~ ~$\omega$~;
    tracker : meta_configuration ~$\Pi$~ ~$\omega$~;
  }.
\end{minted}

As before, we can also define a semantics relating successive steps in an augmented run. Informally,
we have that augmented configuration $C$ transitions to augmented configuration $C'$ upon
execution of line $l$ by process $\pi$ iff the semantics of the implementation dictates
that the base configuration of $C$ transitions to the base configuration of $C'$ upon
execution of line $l$ by process $\pi$, and the tracker of $C'$ is equal to
$\Evolve \; C \; \pi \; l$. As before, we define a predicate $\WF{\mathtt{aug}}{-}$ on augmented
runs classifying those which are well-formed. The tracker of the initial configuration of a
well-formed augmented run is initialized to contain exactly the atomic configuration $(\sigma_0, f)$ where
$f(\pi) = \Idle$ for all $\pi \in \Pi$ and $\sigma_0$ is the initialization of $\omega$. With these definitions, we are now ready to state the main metatheoretic result.

\begin{restatable}[Adequacy]{theorem}{adequacy}
  An implementation is linearizable if and only if the tracker of the final configuration of every augmented
  run admitted by that implementation is inhabited.
\end{restatable}

We can formalize the statement of this theorem in Rocq as follows:

\begin{minted}[escapeinside=~~]{coq}
  Theorem adequacy : linearizable impl ~$\leftrightarrow$~ ~$\forall r$~, ~$\exists$~ ~$\sigma$~ f, (final r).(tracker) ~$\sigma$~ f
\end{minted}

The forward direction of this double implication is \textit{soundness}; that an implementation can only be proved linearizable
if it is indeed linearizable. The reverse direction is \textit{completeness}; that every linearizable implementation can be proved
such. We now provide a pen and paper proof of this theorem, mirroring our mechanization. First, we prove a lemma whose result implies soundness.
For convenience, let $\mathcal{M}(r)$ denote the tracker associated with the final configuration of run $r$ below.

\begin{lemma}\label{lem:soundness}
  Consider an arbitrary augmented run $r$. For any atomic configuration $(\sigma, f)$, if $(\sigma, f) \in \mathcal{M}(r)$, then $(\sigma, f)$ is the final configuration of some linearization of $r$.
\end{lemma}


\begin{proof} We proceed by structural induction on $r$.
    \begin{itemize}
      \item Case $\Initial{c}$: Then, the final configuration of $r$ is simply $c$ itself. By definition, the tracker is
            initialized to contain only $(\sigma_0, f)$, where $\sigma_0$ is the initial state of $\omega$ and
            $f(\pi) = \Idle$ for every $\pi \in \Pi$. This is the final configuration of an initial atomic run, which is clearly a linearization of $r$.
            Thus, this case is complete.
      \item Case $\Step{r'}{\pi}{l}{c}$. Consider any $(\sigma', f') \in \mathcal{M}(r)$. Because we assume that $r$ is well-formed,
      $\tracker{c}$ must result from the evolution of $\mathcal{M}(r)$; that is, we have that $\tracker{c} = \Evolve \; \pi \; l \; \mathcal{M}(r)$. Now case on $l$, considering the kind of line executed.
            \begin{itemize}
              \item Case $\Invoke{\mathit{op}}{\mathit{arg}}$: We have that $\tracker{c} = \Evolve \; \pi \; l \; \mathcal{M}(r)$ is defined to be
                    $\LinearizePending(\EvolveInv(\mathit{op}, \mathit{arg}, \mathit{\pi}))$ Thus, there must exist some $(\sigma, f) \in \mathcal{M}(r)$ and $\vec{\pi} \in \Pi^\ast$ such that
                    $(\sigma, \Sub{\Pending{\mathit{op}}{\mathit{arg}}}{\pi}{f}, \vec{\pi}, \sigma', f') \in \delta^\ast_\omega$. Now proceed by induction on the length of $\vec{\pi}$.
                    \begin{itemize}
                      \item Case $\vec{\pi} = \cdot$: Then, we have that $\sigma' = \sigma$ and $f' = \Sub{\Pending{\mathit{op}}{\mathit{arg}}}{\pi}{f}$. By the outer induction hypothesis, because $(\sigma, f) \in \mathcal{M}(r)$, there must exist some
                            linearization $r_\mathit{atomic}$ of $r$ whose final configuration is $(\sigma, f)$. It then follows that
                            $\Step{r_\mathit{atomic}}{\pi}{(\Invoke{\mathit{op}}{\mathit{arg}})}{(\sigma, \Sub{\Pending{\mathit{op}}{\mathit{arg}}}{\pi}{f})}$ is a linearization of
                            $\Step{r}{\pi}{(\Invoke{\mathit{op}}{\mathit{arg}})}{c}$ as desired.
                      \item Case $\vec{\pi} \circ \pi'$: Then, there must exist some intermediate atomic configuration $(\sigma'', f'')$ where
                            $(\sigma, f, \vec{\pi}, \sigma'', f'') \in \delta^\ast_\omega$, $f''(\pi) = \Pending{\mathit{op}}{\mathit{arg}}$, $(\sigma'', \pi, \mathit{op'}, \mathit{arg'}, \sigma') \in \delta_\omega$, and $f' = \Sub{\Linearized{v}}{\pi}{f''}$. 
                            By induction, exists some linearization $r_\mathit{atomic}'$ of $\Step{r}{\pi}{l}{c}$ with final configuration $(\sigma'', f'')$. Clearly,
                            if $r_\mathit{atomic}'$ is a linearization of $\Step{r}{\pi}{l}{c}$, then so is $\Step{r_\mathit{atomic}'}{\pi'}{\Intermediate}{(\sigma, \Sub{\Linearized{v}}{\pi'}{f''})}$, completing the proof.
                    \end{itemize}
                \item Cases $\Intermediate$, $\Response{v}$: Omitted for brevity.
            \end{itemize}
    \end{itemize}
\end{proof}

We omit the $\Intermediate$ and $\Response{v}$ cases for brevity. They are also very similar to the invocation case,
and are included in the mechanization. Next, we state a lemma whose result implies completeness.

\begin{lemma}\label{lem:completeness}
  Consider any run $r$ and any linearization $r_\mathbf{atomic}$ with final configuration $(\sigma, f)$.
  Then, $(\sigma, f)$ is contained in the tracker of the final configuration of $r$.
\end{lemma}

For brevity, we omit an informal proof of completeness here. However, of course, this proof is included in the mechanization. Before proving the main result, we must deal with one hiccup. The soundness and completeness results refer to \textit{augmented} runs, whereas linearizability
refers to non-augmented runs (i.e. those without a tracker). However, note that augmented runs are isomorphic to implementation runs;
we can \textit{embed} an implementation run into an augmented run, and \textit{project} an augmented run to an implementation run. These operations simply add and remove 
the augmented (ghost state) from a run. Because the evolution of the tracker at every step is deterministic, these operations are mutually inverse.
They also preserve well-formedness in the sense that $\WF{\mathtt{aug}}{\mathtt{embed} \; r}$ is true if $\WF{\mathtt{impl}}{r}$, and dually for projection.

\adequacy*

\begin{proof}
  We prove both directions in turn:
  \begin{itemize}
    \item $(\implies)$ Suppose an implementation is linearizable and consider any augmented run $r$. By assumption, there must exist some linearization $r_\mathbf{atomic}$ of the projection of $r$.
    Note that a linearization of the projection of $r$ is also a linearization of $r$, as their behaviors are the same.
    Then, by \Cref{lem:completeness}, the tracker of $r$ must contain the final configuration of $r_\mathbf{atomic}$, and is thus inhabited.
    \item $(\impliedby)$ Suppose the tracker of every augmented run of an implementation is inhabited. To prove that an implementation is linearizable, it suffices to show that there exists a some linearization of
    every run $r$. Because $\tracker{\final{(\mathtt{embed} \; r)}}$ is inhabited, there must exist some $(\sigma, f) \in \tracker{\final{(\mathtt{embed} \; r)}}$. By \Cref{lem:soundness}, $(\sigma, f)$ is the final configuration of some linearization of $\mathtt{embed} \; r$. 
    Any linearization of an embedding of $r$ is also an linearization of $r$, and so there in turn exists some linearization of $r$, completing the proof.
  \end{itemize}
\end{proof}



\section{Case study}
\label{sec:case-study}
We now apply the tracking technique, and the proof of its metatheory, to develop an end-to-end verified
proof of linearizability for a simple concurrent register. Specifically, we develop and verify
a register supporting \textit{read} and \textit{write} operations,
which is implemented in terms of a base register that only supports read and compare-and-swap (CAS).
One might imagine that this case study is not of practical interest, as all modern processors support word-length read, write, and CAS operations.
However, most architectures support \textit{double-width} compare-and-swap (DWCAS), but \textit{not} double-width writes. This algorithm
allows one to simulate double-width write on modern architectures.

Formally, we take $\Omega = \{ \RWCell \}$, a singleton containing exactly the concurrent register to be implemented. We can then define the \textit{type} of $\RWCell$ as follows:

\begin{itemize}
  \item We define its set of states, $\States{\RWCell}$ to be $\Val$, the set of all possible values. This means that a value of arbitrary type can be stored in the register.
  \item Define $\OP{\RWCell}$ to be $\{\Read, \Write\}$, indicating that the cell should support read and write operations.
  \item The transition relation $\delta_\RWCell$ is described by its action on both operations. $\Read$ simply returns the value currently stored in the cell without mutating the state, and
  $\Write$ updates the state to the value passed as an argument.
\end{itemize}

We define $\overline{\Omega}$, the set of base objects used to implement $\Omega$, to be $\{ \RCASCell \}$.
It is a singleton containing exactly one cell that supports read and CAS operations. Again, we can now
define its type:

\begin{itemize}
  \item We define its set of states, $\Sigma_\RCASCell$ to be $\Val$, the set of all possible values.
  \item Define $\OP{\RWCell}$ to be $\{\Read, \CAS\}$.
  \item We encode the transition relation $\delta_\RCASCell$ as follows. The semantics for $\Read$ remains unchanged. When provided a pair $(\mathit{cur}, \mathit{new})$, $\CAS$ will compare \textit{cur} to the value currently stored in the cell. If they are equal, then \textit{new} is written to the cell. Otherwise, the operation fails, and the state is unchanged.
\end{itemize}

We can implement $\Omega$ with $\overline{\Omega}$ as described by pseudocode in \Cref{fig:rwcas}. A read operation
performed on the $\RWCell$ simply performs a read on the $\RCASCell$ and returns the result. A $\Write(y)$ operation first reads the value currently stored in the cell, stores it in $x$, and performs $\CAS(x, y)$.

\begin{figure}[t]

  \caption{Implementation of a read-write register using read and CAS operations}
  
  \begin{minipage}[t]{0.48\textwidth}
  \begin{algorithm}[H]
  \DontPrintSemicolon
  \SetKwProg{Fn}{Procedure}{}{}
  \Fn{\textbf{Read()}}{

    $x \leftarrow \Read(\RCASCell)$\;
    \Return $x$\;
  }
  \end{algorithm}
  \end{minipage}\hfill
  \begin{minipage}[t]{0.48\textwidth}
  \begin{algorithm}[H]
  \DontPrintSemicolon
  \SetKwProg{Fn}{Procedure}{}{}
  \Fn{\textbf{Write($y$)}}{
    $x \leftarrow \Read(\RCASCell)$\;
    $\CAS$($\RCASCell$, $x$, $y$)$\;$\;
    \Return $\top$\;
  }
  \end{algorithm}
  \end{minipage} 

\label{fig:rwcas}

\end{figure}

It is not immediately clear why this implementation is linearizable. If another process $\pi'$ writes a different value to the cell between lines 1 and 2 of a $\Write$ by $\pi$, then the $\CAS$ may fail, and the write will never \textit{physically} take effect. Intuitively, however, this means that the write by $\pi$
could have linearized immediately before the write of $\pi'$. We formalize this intuition in our mechanization,
providing a verified end-to-end proof of linearizability for this implementation.

\section{Mechanization}
\label{sec:mechanization}
We have mechanized all of the results described in the Rocq proof assistants,
encoding all data structures as formulated in the above snippets. The
mechanization of the informal adequacy proofs follows the structure
outlined in \Cref{sec:tracker}. The mechanized proof of
linearizability of the concurrent register is quite complex, likely
due to the lack of semantics resulting from an untyped language.
In the mechanization, we axiomatize functional extensionality so that
extenstionally equal functions representing maps are indeed equal.
We also depend on the \texttt{stdpp} package, from which we employ
its finite maps. The mechanization amounts to about 3250 lines of
Rocq total.


\section{Conclusion}
\label{sec:conclusion}
We have presented a mechanized metatheory of the \textit{metaconfigurations tracking technique},
verifying its soundness and completeness for proving linearizability. To this end, we formalized
a semantics for an asychronous concurrent system, linearizability itself, and the adequacy of the
technique. This allows for the construction of the first \textit{end-to-end} linearizability proofs
using this tracking technique. We demonstrate this through a mechanization of an
end-to-end proof of a concurrent register whose implementation is already used in real-world systems.



\bibliography{main}

\appendix

\end{document}